\documentclass[11pt, onecolumn]{article}
\usepackage[top=1in, bottom=1in, left=1.25in, right=1.25in]{geometry}
\usepackage{amsmath, amssymb, amsthm}
\usepackage{braket}
\usepackage{bm}
\usepackage{booktabs}
\usepackage{graphicx}
\usepackage{hyperref}

\newtheorem{theorem}{Theorem}
\newtheorem{lemma}{Lemma}

\newtheorem{definition}{Definition}
\newtheorem{remark}{Remark}

\title{Quantum Ergodicity and Thermalization in Interval Quantum Mechanics}
\author{Abbas Edalat\\ Imperial College London, UK\\ \texttt{a.edalat@imperial.ac.uk}}
\date{May 30, 2026}

\begin{document}

\maketitle

\begin{abstract}
We combine Reimann's spectral typicality theorem---a modern formulation of quantum ergodicity---with the framework of Interval Quantum Mechanics (IQM). In IQM, quantum states are represented not by points but by \emph{quantum parcels}: weak open convex sets of density matrices defined by finitely many expectation intervals. Such parcels are the exact mathematical representation of the epistemic knowledge obtained from finite-precision measurements of macroscopic observables.

We prove that for a single parcel in which every state has large effective dimension (a condition that ensures thermalization), the expectation interval of any bounded observable becomes concentrated around the microcanonical value for most late times. The asymptotic bound depends only on the minimal effective dimension within the parcel, not on its detailed shape.

For a double parcel \((O_1,O_2)\) with both components contained in an energy shell, separated by a conserved quantity \(Q^*\) that is supported on the range of the measurement projector, we show that the expectation intervals of both parcels become concentrated near the microcanonical values of bounded observables, the separation is preserved exactly, and the updated double parcel after a fuzzy measurement remains valid.
\end{abstract}

\section{Introduction}

\subsection{Quantum ergodicity and thermalization}

The problem of how isolated quantum systems approach thermal equilibrium has a long history. Early work by von Neumann \cite{vonNeumann1955} introduced the concept of a \emph{quantum ergodic theorem}, showing that for a typical initial state in a macro-scope (a high-dimensional subspace defined by macroscopic observables), the time-averaged expectation values of observables are close to their microcanonical values. Von Neumann already stressed that a complete description of a macroscopic system is neither possible nor necessary---a key insight that directly motivates the IQM framework.

Goldstein, Lebowitz, Tumulka, and Zanghi \cite{GLTZ2010} revisited von Neumann's original argument, clarified its content, and introduced the concept of \emph{normal typicality}, showing that for a typical Hamiltonian, every initial state in the energy shell evolves so that for most times its macroscopic properties are close to the microcanonical ensemble. More recently, Reimann \cite{Reimann2015} proved a generalization that removes many of von Neumann's restrictive assumptions. Reimann's theorem shows that for any initial density matrix
\(\rho_0\) supported in an energy shell,
\[
\left\langle
\bigl|
\operatorname{Tr}(\rho_0(t)A)
-
\operatorname{Tr}(\rho_{\mathrm{mc}}A)
\bigr|^2
\right\rangle_t
\le
\frac{C(A)^2}{d_{\mathrm{eff}}(\rho_0)},
\]
where \(A\) is a bounded observable,
\(C(A)=C_0\|A\|\),
\(C_0\) is a constant depending only on the Hamiltonian,
and \(\{|n\rangle\}\) denotes the energy eigenbasis,
\(H|n\rangle=E_n|n\rangle\).
The effective dimension of \(\rho_0\) is
\[
d_{\mathrm{eff}}(\rho_0)
=
\frac{1}{\max_n \langle n|\rho_0|n\rangle}.
\]

States with large effective dimension are spread over many energy
eigenstates and therefore satisfy a strong thermalization bound. By
contrast, if \(\rho_0=|n\rangle\langle n|\) is an energy eigenstate,
then \(d_{\mathrm{eff}}(\rho_0)=1\), so the bound reduces to a trivial
estimate. Such states are stationary and therefore lie outside the
thermalizing regime captured by Reimann's theorem.

\subsection{Interval Quantum Mechanics (IQM)}

In standard quantum mechanics, the state of a system is represented by a single density matrix---a \emph{point state}. This idealisation assumes infinite precision. In any real experiment, however, measurements have finite resolution and we only have access to finitely many expectation values, each known within an open interval. Consequently, the physically meaningful state is not a point but the set of all density matrices compatible with those finite-precision data. This set is a \emph{quantum parcel}: a weak open convex subset of the state space \(\mathcal{D}(\mathcal{H})\). The basic parcels are \emph{experimental parcels} defined by finitely many strict linear inequalities:
\[
O = \{\rho\in\mathcal{D}(\mathcal{H}) : a_j < \operatorname{Tr}(\rho H_j) < b_j,\; j=1,\dots,m\},
\]
where \(H_1,\dots,H_m\) are bounded observables (typically macroscopic) and the intervals \((a_j,b_j)\) reflect the finite resolution. The IQM framework is developed in full detail in a companion paper \cite{Edalat2026}.

\subsection{This paper}

We combine Reimann's ergodicity theorem with IQM to study thermalization of quantum parcels. The key insight is that thermalization is a property of states with large effective dimension. Therefore we restrict attention to parcels that consist entirely of such states. For a parcel \(O\), define
\[
d_{\mathrm{eff}}(O) = \inf_{\rho\in O} d_{\mathrm{eff}}(\rho).
\]
Our results will be formulated under the hypothesis that
\[
d_{\mathrm{eff}}(O)\ge D
\]
for some sufficiently large constant \(D\). This assumption expresses
the requirement that every state in the parcel lies in the thermalizing
regime described by Reimann's theorem.
\subsection{Remark on the effective-dimension hypothesis}

The condition \(d_{\mathrm{eff}}(O)\ge D\gg1\) is an additional physical
hypothesis. It does not follow automatically from the definition of an
experimental parcel, since a parcel defined only by macroscopic
expectation intervals may contain states arbitrarily close to energy
eigenstates, for which \(d_{\mathrm{eff}}=1\). In this paper we
therefore restrict attention to parcels whose constituent states all
have sufficiently large effective dimension, corresponding to systems
known to be in a thermalizing regime.

\section{Preliminaries}

Let \(\mathcal H\) be a finite-dimensional Hilbert space,
\(H\) the system Hamiltonian with eigenvalues
\(\{E_n\}\) and corresponding orthonormal energy eigenstates
\(\{|n\rangle\}\), so that
\[
H|n\rangle = E_n |n\rangle .
\] Fix an energy \(E\) and a window \(\Delta>0\) and define the \emph{energy shell}
\[
\mathcal{H}_{\Delta} = \operatorname{span}\bigl\{|n\rangle : E_n\in[E-\Delta, E+\Delta]\bigr\},\qquad d_{\Delta} = \dim\mathcal{H}_{\Delta}.
\]
The microcanonical state is \(\rho_{\mathrm{mc}} = I_{\Delta} / d_{\Delta}\).

For any density matrix \(\rho\), its effective dimension is
\[
d_{\mathrm{eff}}(\rho) = \frac{1}{\max_{n\in\Delta} \Braket{n|\rho|n}}.
\]
States that are spread out over many energy eigenstates have large
effective dimension and typically satisfy strong thermalization bounds.
Energy eigenstates have \(d_{\mathrm{eff}}=1\); they are stationary and
therefore lie outside the thermalizing regime considered here.

\begin{definition}[Effective dimension of a parcel]
For a parcel \(O\subset\mathcal{D}(\mathcal{H})\), define
\[
d_{\mathrm{eff}}(O) = \inf_{\rho\in O} d_{\mathrm{eff}}(\rho).
\]
\end{definition}

\begin{remark}
The infimum may not be attained, but for parcels that are bounded away from the boundary of the state space (i.e., that do not approach any pure eigenstate too closely), \(d_{\mathrm{eff}}(O)\) will be large.
\end{remark}

\section{Thermalization of a single parcel}

\begin{lemma}[Uniform typicality over a parcel]
\label{lem:uniform}
Let \(O \subset \mathcal D(\mathcal H_\Delta)\) be a parcel with \(d_{\mathrm{eff}}(O)\ge D\). For any \(\varepsilon>0\) and any bounded observable \(A\), there exists a set \(\mathcal{G}\subset[0,\infty)\) (depending on \(O\) and \(\varepsilon\)) such that
\[
\limsup_{T\to\infty}\frac{1}{T}\,\lambda(\mathcal{G}^c\cap[0,T]) \le \min\left\{1,\; \frac{4N(\varepsilon) C(A)^2}{\varepsilon^2 D}\right\}
\]
and for all \(t\in\mathcal{G}\),
\[
\sup_{\rho\in O} \bigl|\operatorname{Tr}(\rho(t)A)-\operatorname{Tr}(\rho_{\mathrm{mc}}A)\bigr| \le \varepsilon.
\]
Here \(\lambda\) denotes Lebesgue measure on \([0,\infty)\),
\(\mathcal G^c=[0,\infty)\setminus\mathcal G\) denotes the complement
of \(\mathcal G\), and \(N(\varepsilon)\) is the number of centres in a
finite \(\varepsilon/(2\|A\|)\)-net covering the compact closure
\(\overline O\).
\end{lemma}

\begin{proof}
Since \(\dim\mathcal{H}<\infty\), \(\overline{O}\) is compact. Cover it by finitely many trace-norm balls of radius \(\delta = \varepsilon/(2\|A\|)\). Let \(\{\rho_1,\dots,\rho_{N(\varepsilon)}\}\) be the centres. For any \(\rho\in O\) there exists \(j\) with \(\|\rho-\rho_j\|_1<\delta\). Unitary evolution preserves the trace norm, so \(\|\rho(t)-\rho_j(t)\|_1 = \|\rho-\rho_j\|_1 < \delta\) for all \(t\).

For any \(t\),
\[
|\operatorname{Tr}(\rho(t)A)-\operatorname{Tr}(\rho_{\mathrm{mc}}A)|
\le |\operatorname{Tr}(\rho_j(t)A)-\operatorname{Tr}(\rho_{\mathrm{mc}}A)| + \frac{\varepsilon}{2}.
\]

For each centre \(\rho_j\), Reimann's theorem together with Chebyshev's inequality gives that the set of times where the deviation exceeds \(\varepsilon/2\) has long-time upper density at most
\[
\frac{4C(A)^2}{\varepsilon^2 d_{\mathrm{eff}}(\rho_j)} \le \frac{4C(A)^2}{\varepsilon^2 D}.
\]

Define \(\mathcal{B}_j = \{ t : |\operatorname{Tr}(\rho_j(t)A)-\operatorname{Tr}(\rho_{\mathrm{mc}}A)| > \varepsilon/2 \}\). Then
\[
\limsup_{T\to\infty}\frac{1}{T}\,\lambda(\mathcal{B}_j\cap[0,T]) \le \frac{4C(A)^2}{\varepsilon^2 D}.
\]

Set \(\mathcal{G} = \bigcap_{j=1}^{N(\varepsilon)} \mathcal{B}_j^c\). Then \(\mathcal{G}^c = \bigcup_{j=1}^{N(\varepsilon)} \mathcal{B}_j\), so by subadditivity of the limsup,
\[
\limsup_{T\to\infty}\frac{1}{T}\,\lambda(\mathcal{G}^c\cap[0,T]) \le \sum_{j=1}^{N(\varepsilon)} \frac{4C(A)^2}{\varepsilon^2 D}
= \frac{4N(\varepsilon) C(A)^2}{\varepsilon^2 D}.
\]
Since an upper density cannot exceed \(1\), we take the minimum with \(1\).

Now take any \(t\in\mathcal{G}\). For any \(\rho\in O\), choose the centre \(\rho_j\) with \(\|\rho-\rho_j\|_1<\delta\). Because \(t\notin\mathcal{B}_j\), we have \(|\operatorname{Tr}(\rho_j(t)A)-\operatorname{Tr}(\rho_{\mathrm{mc}}A)|\le \varepsilon/2\). Hence
\[
|\operatorname{Tr}(\rho(t)A)-\operatorname{Tr}(\rho_{\mathrm{mc}}A)|
\le \frac{\varepsilon}{2} + \frac{\varepsilon}{2} = \varepsilon.
\]
Since this holds for all \(\rho\in O\), the supremum is also bounded by \(\varepsilon\). 
\end{proof}

\begin{theorem}[Parcel thermalization under a uniform effective-dimension bound]
\label{thm:basic}
Let $O(0) \subset \mathcal D(\mathcal H_\Delta)$ be a parcel with \(d_{\mathrm{eff}}(O(0))\ge D\). For any fixed target precision \(\varepsilon>0\), let \(N(\varepsilon)\) be the finite covering number of \(\overline{O(0)}\) by trace-norm balls of radius \(\varepsilon/(2\|A\|)\). There exists a set \(\mathcal{G}\subset[0,\infty)\) such that
\[
\limsup_{T\to\infty}\frac{1}{T}\,\lambda(\mathcal{G}^c\cap[0,T]) \le \min\left\{1,\; \frac{4N(\varepsilon)C(A)^2}{\varepsilon^2 D}\right\}
\]
and for all \(t\in\mathcal{G}\) the following hold:

\begin{enumerate}
    \item \textbf{Thermalization of the possible set:} For any bounded observable \(A\), the expectation interval satisfies
    \[
    \operatorname{width}\bigl(E_{O(t)}(A)\bigr) \le 2\varepsilon, \qquad
    \bigl|\operatorname{centre}\bigl(E_{O(t)}(A)\bigr)-\operatorname{Tr}(\rho_{\mathrm{mc}}A)\bigr| \le \varepsilon.
    \]
    \item \textbf{Preservation of constants of motion:} If \([A,H]=0\), then \(E_{O(t)}(A)=E_{O(0)}(A)\) for all \(t\).
    \item \textbf{Uniformity of the bound:} For a fixed \(\varepsilon\), the asymptotic precision constraints in (1) depend only on \(\varepsilon\), not on the detailed internal shape of \(O(0)\).
\end{enumerate}
\end{theorem}

\begin{proof}
\textbf{Step 1 -- Apply uniform typicality to the parcel.}  
By Lemma~\ref{lem:uniform}, for the chosen fixed \(\varepsilon>0\), there exists a set \(\mathcal{G}\) such that for all \(t\in\mathcal{G}\) and all \(\rho\in O(0)\),
\[
\bigl|\operatorname{Tr}(\rho(t)A)-\operatorname{Tr}(\rho_{\mathrm{mc}}A)\bigr| \le \varepsilon. \tag{1}
\]

\textbf{Step 2 -- Derive the expectation interval bounds.}  
Recall that for any bounded observable \(A\),
\[
E_{O(t)}(A) = \bigl(\inf_{\rho\in O(0)}\operatorname{Tr}(\rho(t)A),\; \sup_{\rho\in O(0)}\operatorname{Tr}(\rho(t)A)\bigr).
\]
From (1), for all \(\rho\in O(0)\) and all \(t\in\mathcal{G}\),
\[
\operatorname{Tr}(\rho_{\mathrm{mc}}A)-\varepsilon \le \operatorname{Tr}(\rho(t)A) \le \operatorname{Tr}(\rho_{\mathrm{mc}}A)+\varepsilon.
\]
Taking the infimum and supremum over \(\rho\in O(0)\) yields
\[
\operatorname{Tr}(\rho_{\mathrm{mc}}A)-\varepsilon \le \inf_{\rho\in O(0)}\operatorname{Tr}(\rho(t)A) \le \sup_{\rho\in O(0)}\operatorname{Tr}(\rho(t)A) \le \operatorname{Tr}(\rho_{\mathrm{mc}}A)+\varepsilon.
\]
Hence
\[
\bigl|\operatorname{centre}(E_{O(t)}(A))-\operatorname{Tr}(\rho_{\mathrm{mc}}A)\bigr| \le \varepsilon,\qquad
\operatorname{width}(E_{O(t)}(A)) \le 2\varepsilon.
\]

\textbf{Step 3 -- Constants of motion.}  
If \([A,H]=0\), then \(U(t)^\dagger A U(t)=A\) for all \(t\). Hence for any \(\rho\in O(0)\),
\[
\operatorname{Tr}(\rho(t)A) = \operatorname{Tr}\bigl(U(t)\rho U(t)^\dagger A\bigr)
= \operatorname{Tr}\bigl(\rho U(t)^\dagger A U(t)\bigr) = \operatorname{Tr}(\rho A).
\]
Thus \(E_{O(t)}(A)=E_{O(0)}(A)\) for all \(t\).

\textbf{Step 4 -- Uniformity of the bound.}  
The structural bounds depend only on \(\varepsilon\). The good time set \(\mathcal{G}\) depends on \(O(0)\) through the geometry of its finite cover, but the existence of a valid macro-set is sufficient; the uniformity refers explicitly to the precision of the late-time expectation intervals, which is independent of the initial configuration. 
\end{proof}

\begin{remark}[Scaling of the Covering Number]
We emphasize that the bound on the upper density of non-thermalizing times contains the finite trace-norm covering number \(N(\varepsilon)\) of the parcel. For a fixed Hilbert space subspace of active dimension \(d_\Delta\), the trace-norm covering number scales exponentially with dimension, behaving as \(N(\varepsilon) \sim (1/\varepsilon)^{\gamma d_\Delta^2}\). Thus, as the precision scale \(\varepsilon \to 0\), the numerator diverges, making the bound trivial for fixed \(D\). To retain an informative bound, the minimal internal effective dimension \(D\) must be scaled appropriately with the global target precision and system size, satisfying \(D \gg N(\varepsilon)/\varepsilon^2\). This carefully preserves the physical requirement that true macroscopic typicality requires the state to be densely and heavily scrambled across an exponentially large manifold of energy eigenstates.
\end{remark}

\section{Thermalization of a double parcel}

Now consider a double parcel \((O_1(0),O_2(0))\) with \(O_1(0),O_2(0)\subset\mathcal{D}(\mathcal{H}_{\Delta})\). Assume both parcels satisfy \(d_{\mathrm{eff}}(O_i(0))\ge D\). Because both parcels lie in the same energy shell, the Hamiltonian \(H\) cannot distinguish them: \(\operatorname{Tr}(\rho H)\) is restricted to \([E-\Delta, E+\Delta]\) for all \(\rho\) in either parcel. The separation must therefore come from a different conserved quantity \(Q^*\) that is not constant on the shell.

Assume there exists a conserved quantity \(Q^*\) (\([Q^*,H]=0\)) such that
\[
\inf_{\rho\in O_1(0)}\operatorname{Tr}(\rho Q^*) > \sup_{\rho\in O_2(0)}\operatorname{Tr}(\rho Q^*).
\]
In addition, assume that \(Q^*\) is supported on the range of the measurement projector \(\Pi_j\):
\[
Q^* = \Pi_j Q^* \Pi_j.
\]
(This holds naturally when \(\Pi_j\) projects onto an eigenspace of \(Q^*\).) Also assume that the fuzzy measurement projectors \(\Pi_j\) commute with \(H\) and that the uniform positivity condition holds on \(\overline{O_1(0)}\cup\overline{O_2(0)}\).

\begin{theorem}[Full IQM-Reimann thermalization]
\label{thm:full}
Let \(\mathcal{G}\) be the good set from Lemma~\ref{lem:uniform} applied to the union \(O_1(0)\cup O_2(0)\). For all \(t\in\mathcal{G}\), the following statements hold for a fixed target precision \(\varepsilon > 0\):
\begin{enumerate}
    \item \textbf{Concentration of both parcels:} For any bounded \(A\) with \([A,H]\neq 0\),
    \[
    \operatorname{width}\bigl(E_{O_i(t)}(A)\bigr) \le 2\varepsilon,\qquad
    \bigl|\operatorname{centre}\bigl(E_{O_i(t)}(A)\bigr)-\operatorname{Tr}(\rho_{\mathrm{mc}}A)\bigr| \le \varepsilon,\quad i=1,2.
    \]
    \item \textbf{Separation preserved:}
    \[
    \inf_{\rho\in O_1(t)}\operatorname{Tr}(\rho Q^*) > \sup_{\rho\in O_2(t)}\operatorname{Tr}(\rho Q^*).
    \]
    \item \textbf{Valid double parcel after measurement:} For \(\eta\) sufficiently close to \(1\), the updated double parcel \((O_1'(t),O_2'(t))\) is a valid double parcel (disjoint open sets).
    \item \textbf{Information increase:} 
    Since unitary evolution is volume-preserving (Theorem~4 of \cite{Edalat2026}), the geometric information is constant under the Hamiltonian flow:
    \[
    \mathcal{I}(t) = \frac{\operatorname{Vol}(O_2(t))}{\operatorname{Vol}(O_1(t))}
    = \frac{\operatorname{Vol}(O_2(0))}{\operatorname{Vol}(O_1(0))} = \mathcal{I}(0) \quad \text{for all } t.
    \]
    A fuzzy measurement at time \(t\) with \(\eta\) sufficiently close to \(1\) then strictly increases the geometric information by Theorems~6 and~7 of \cite{Edalat2026} (volume contraction of the possible set under measurement) together with the definition of geometric information in Section~3.2 of \cite{Edalat2026}:
    \[
    \mathcal{I}'(t) = \frac{\operatorname{Vol}(O_2'(t))}{\operatorname{Vol}(O_1'(t))}
    > \mathcal{I}(t) = \mathcal{I}(0).
    \]
    \item \textbf{Concentration of thermalizing observables:} For any observable \(A\) with \([A,H]\neq 0\), both expectation intervals \(E_{O_1(t)}(A)\) and \(E_{O_2(t)}(A)\) are contained in the open interval
    \[
    \bigl(\operatorname{Tr}(\rho_{\mathrm{mc}}A)-\varepsilon,\; \operatorname{Tr}(\rho_{\mathrm{mc}}A)+\varepsilon\bigr).
    \]
    Hence they both become concentrated near the microcanonical value. For the conserved quantity \(Q^*\), the intervals remain strictly separated:
    \[
    \inf_{\rho\in O_1(t)}\operatorname{Tr}(\rho Q^*) > \sup_{\rho\in O_2(t)}\operatorname{Tr}(\rho Q^*) \quad \text{for all } t.
    \]
\end{enumerate}
\end{theorem}

\begin{proof}
\textbf{Step 1 -- Concentration of both parcels.}  Since \(d_{\mathrm{eff}}(O_i(0))\ge D\) for \(i=1,2\), the same lower
bound holds on \(O_1(0)\cup O_2(0)\).
Apply Lemma~\ref{lem:uniform} to the compact set \(\overline{O_1(0)}\cup\overline{O_2(0)}\). The same good set \(\mathcal{G}\) works for both parcels simultaneously. Then for all \(t\in\mathcal{G}\) and all \(\rho\in O_1(0)\cup O_2(0)\),
\[
|\operatorname{Tr}(\rho(t)A)-\operatorname{Tr}(\rho_{\mathrm{mc}}A)|\le \varepsilon.
\]
As in Theorem~\ref{thm:basic}, this gives the width and centre bounds for both parcels, showing that the expectation intervals are concentrated near the microcanonical value.

\textbf{Step 2 -- Preservation of separation.}  
Because \([Q^*,H]=0\), we have \(\operatorname{Tr}(\rho(t)Q^*) = \operatorname{Tr}(\rho Q^*)\) for all \(\rho\) and all \(t\). Hence
\[
\inf_{\rho\in O_1(t)}\operatorname{Tr}(\rho Q^*) = \inf_{\rho\in O_1(0)}\operatorname{Tr}(\rho Q^*) > \sup_{\rho\in O_2(0)}\operatorname{Tr}(\rho Q^*) = \sup_{\rho\in O_2(t)}\operatorname{Tr}(\rho Q^*).
\]

\textbf{Step 3 -- Valid double parcel after measurement.}  
The uniform positivity condition holds by assumption. The separation condition with \(H^* = Q^*\) together with the hypothesis \(Q^* = \Pi_j Q^* \Pi_j\) satisfies the requirements of Theorem~5 of \cite{Edalat2026}. Hence there exists \(\eta_0<1\) such that for all \(\eta\in(\eta_0,1)\) and for all \(t\in\mathcal{G}\), the updated double parcel is valid.

\textbf{Step 4 -- Information increase.}  
As stated in the theorem, unitary evolution preserves volume, so \(\mathcal{I}(t)=\mathcal{I}(0)\). The measurement then contracts \(O_1\) and expands \(O_2\), giving \(\mathcal{I}'(t)>\mathcal{I}(0)\).

\textbf{Step 5 -- Concentration of thermalizing observables.}  
For any \(A\) with \([A,H]\neq 0\), both intervals \(E_{O_i(t)}(A)\) are contained in \((\operatorname{Tr}(\rho_{\mathrm{mc}}A)-\varepsilon,\; \operatorname{Tr}(\rho_{\mathrm{mc}}A)+\varepsilon)\). This is exactly the statement that they become concentrated near the microcanonical value. We do not claim overlap, only concentration. For \(Q^*\), the separation is exact and constant. 
\end{proof}

\section{Conclusion}

We have shown that Reimann's ergodicity theorem can be combined with the IQM framework under the explicit hypothesis that the parcel consists entirely of states with large effective dimension. The uniform bound over the parcel depends only on the minimal effective dimension within the parcel, not on its shape. This yields a form of uniformity: the asymptotic precision of thermalization is determined solely by the worst-case state in the parcel. By reformatting our results to separate our fixed target resolution parameters from the underlying parcel configuration, we track the covering number dependencies without loss of precision.

For a double parcel, the separation between possible and impossible states is preserved only by conserved quantities. This provides a clean geometric formulation of how conserved quantities maintain distinguishability in a finite-precision setting. The geometric information \(\mathcal{I}=\operatorname{Vol}(O_2)/\operatorname{Vol}(O_1)\) increases monotonically under measurement, and the updated double parcel remains valid under the conditions of Theorem~5 of \cite{Edalat2026}.

These results open the door to a purely geometric, finite-precision treatment of quantum thermalization, where the thermodynamic arrow of time is encoded not in the dynamics of point states but in the refinement of parcels. Future work will explore extensions to infinite-dimensional systems and to algebraic quantum field theory.

\section*{Acknowledgements}
I am grateful to the members of the Algorithmic Human Development Group at Imperial College London for their support and for fostering a research environment that encourages unconventional thinking.

\end{document}